\documentclass[atmp]{ipart_v1}

\Vol{}
\Issue{}
\Year{2020}
\firstpage{1}

\usepackage{t1enc}
\usepackage[latin1]{inputenc}
\usepackage[english]{babel}

\usepackage{amsthm}
\usepackage{yfonts}

\usepackage{bbm}
\usepackage{bm}
\usepackage{mathrsfs}

\newcommand{\be}[0]{\begin{equation}}
\newcommand{\ee}[0]{\end{equation}}

\numberwithin{equation}{section}

\theoremstyle{plain}
\newtheorem{theorem}{Theorem}[section]

\theoremstyle{definition}
\newtheorem{defn}[theorem]{Definition}

\def\tensors3{T^3(E^3)}
\def\lintensors3{\mbox{Lin}(T^3(E^3))}
\begin{document}

\title[Ternary algebras]{Ternary algebras associated with irreducible tensor representations of SO(3) and quark model}
\author[Viktor Abramov]{Viktor Abramov}

\begin{abstract}
We show that each irreducible tensor representation of weight 2 of the rotation group of three-dimensional space in the space of rank 3 covariant tensors gives rise to an associative algebra with unity. We find the algebraic relations that the generators of these algebras must satisfy. Part of these relations has a form of binary relations and another part has a form of ternary relations. The structure of ternary relations is based on the cyclic group $\mathbb Z_3$ and a primitive cubic root of unity $q=\exp(2\pi i/3)$. The subspace of each algebra spanned by the triple products of generators is 5-dimensional and it is the space of an irreducible tensor representation of weight 2 of the rotation group $\mbox{SO}(3)$. We define a Hermitian scalar product in this 5-dimensional subspace and construct an orthonormal basis for it. Then we find the representation matrix of an infinitesimal rotation. We show that constructed algebras with binary and ternary relations can have applications in the quark model and Grand Unification Theories.
\end{abstract}

\maketitle
\section{Introduction and summary}
According to modern concepts of particle physics, the most fundamental particles that do not have an internal structure, that is, are indivisible, are quarks. The main quantum characteristics of quarks are fractional electric charge, fractional baryon number, half-integer spin, color charge, and flavor. Three generations of quarks are known, the first generation is "up" and "down" quarks, the second generation is "strange" and "charmed" quarks and the third generation is "top" and "bottom" quarks. Three quarks (or antiquarks) can combine into one particle, which is called a baryon. The color charge of a whole combination must be zero, that is, the allowed combinations of quarks are colorless. Quarks can also combine in two quark-antiquark, resulting in mesons. The most striking property of the quark model, which was established experimentally, is the so-called quark confinement, that is, a quark cannot be isolated and we can only observe them only in the combinations described above.

Fermions obey the well-known Pauli exclusion principle, which states that there cannot be two fermions in a quantum system with identical quantum characteristics. In the papers \cite{Kerner:Phys_Atom_Nucl_2017},\cite{Kerner:Universe_2019}, the author proposes, on a base of the properties of the quark model, that in the case of quarks, an analog of the Pauli exclusion principle could be formulated as follows: three quarks with completely identical quantum characteristics cannot coexist inside a nucleon. From this formulation, the author then obtains the algebraic properties of a wave function
\begin{eqnarray}
&&\Phi(x_i,x_j,x_k)+\Phi(x_j,x_k,x_i)+\Phi(x_k,x_i,x_j)\nonumber\\
&&\qquad\qquad +\Phi(x_k,x_j,x_i)+\Phi(x_j,x_i,x_k)+\Phi(x_i,x_k,x_i)=0,\label{intr:properties of wave function}
\end{eqnarray}
where $|x_1>,|x_2>,|x_3>$ are three different states and $i,j,k$ in (\ref{intr:properties of wave function}) can be any combination of integers $1,2,3$. In particular, a wave function $\Phi(x_i,x_j,x_k)$, which has $\mathbb Z_3$ symmetry
\begin{equation}
\Phi(x_i,x_j,x_k)=q\;\Phi(x_j,x_k,x_i)=q^2\;\Phi(x_k,x_i,x_j),
\end{equation}
where $q=\exp{(2\pi i/3)}$ is the primitive cube root of unity, satisfies equation (\ref{intr:properties of wave function}). A wave function with such symmetry justifies the introduction and study of a unital associative algebra generated by $\theta^1,\theta^2,\ldots,\theta^n$ that obey ternary relations
\begin{equation}
\theta^i\theta^j\theta^k=q\;\theta^j\theta^k\theta^i=q^2\;\theta^k\theta^i\theta^j.
\label{intr:ternary relations with q}
\end{equation}
From ternary relations (\ref{intr:ternary relations with q}) it follows that for any triple of integers $i,j,k$ from the set $\{1,2,3\}$ the generators of an algebra satisfy the relations
\begin{equation}
\theta^i\theta^j\theta^k+\theta^j\theta^k\theta^i+\theta^k\theta^i\theta^j=0,\;(\theta^i)^3=0
\end{equation}
and a product of any four generators vanishes. Here we can see that algebra with ternary relations (\ref{intr:ternary relations with q}) can be considered as 3rd order analog of Grassmann algebra. Indeed  generators $\xi^1,\xi^2,\ldots,\xi^n$ of Grassmann algebra satisfy the relations
$$
\xi^i\xi^j+\xi^j\xi^i=0,\;(\xi^i)^2=0.
$$
Algebra with relations (\ref{intr:ternary relations with q}), its various generalizations and applications was studied in \cite{Abramov_Kerner_LeRoy:J_Math_Phys_1997},\cite{Abramov_Kerner:J_Math_Phys_2000},
\cite{Abramov_Kerner_Liivapuu:Springer_2020},
\cite{Kerner:Acad_Sci_Paris_1991},\cite{Kerner:J_Math_Phys_1992},\cite{Kerner:Phys_Atom_Nucl_2017}.

In the present paper, we establish a connection between an algebra with ternary relations (\ref{intr:ternary relations with q}) and irreducible tensor representations of weight 2 of the rotations group of three-dimensional space $\mbox{SO}(3)$. In other words, we show that an algebra with ternary relations (\ref{intr:ternary relations with q}) naturally arises in tensor representations of the rotations group. One can naturally associate to three-dimensional Euclidean space $E$ the vector space of complex-valued covectors $T^1(E)$. A basis $e_1,e_2,e_3$ for $E$ gives rise for the dual basis $\vartheta^1,\vartheta^2,\vartheta^3$ in $T^1(E)$, where $\vartheta^i(e_j)=\delta^i_j$. Any covector can now be written in the form $t_i\vartheta^i$, where $t_i\in\mathbb C$ are components of complex-valued rank 1 covariant tensor. Basis covectors $\vartheta^1,\vartheta^2,\vartheta^3$ generate the tensor algebra of complex-valued covariant tensors $T_\vartheta(E)$, where $T_\vartheta(E)=\oplus_a\,T^a_\vartheta(E)$ and $T^a_\vartheta(E)$ is a subspace of rank $a$ covariant tensors, i.e. an element $t(\vartheta)$ of $T^a_\vartheta(E)$ has the form
$$
t(\vartheta)=t_{i_1i_2\ldots i_a}\vartheta^{i_1}\otimes\vartheta^{i_2}\otimes\ldots\otimes\vartheta^{i_a}.
$$
A rotation $e_i\to g^j_i\,e_j$ in $E$ induces the linear transformation $R_g$ in $T^a_\vartheta(E)$ by means of the formula
\begin{equation}
R_g(\vartheta^{i_1}\otimes\vartheta^{i_2}\otimes\ldots\otimes\vartheta^{i_a})=
  {\tilde g}_{j_1}^{i_1}{\tilde g}_{j_2}^{i_2}\ldots {\tilde g}_{j_a}^{i_a}\;
          \vartheta^{j_1}\otimes \vartheta^{j_2}\otimes\ldots\otimes\vartheta^{j_a},
\label{intr:representation of SO(3)}
\end{equation}
where $\tilde g$ is the inverse matrix of $g$. Hence we have a tensor representation of the rotations group.

Of particular interest to us are the irreducible tensor representations of the group of rotations in the space of rank 3 tensors. As is known, in this space the rotations group $\mbox{SO}(3)$ has representations of weight 0,1,2,3. Representation of weight 0 is an irreducible representation of the rotations group in the one-dimensional space of skew-symmetric tensors of rank 3. It is interesting to note that this space is highest dimensional non-trivial space of the Grassmann algebra of three-dimensional space. Of particular interest to us are the irreducible representations of weight 2. In \cite{Gelfand:2018} it is proved that the space of representations of weight 2 is the 10-dimensional space of traceless (with respect to any pair of subscripts) rank 3 tensors, which also satisfy
$$
t_{ijk}+t_{jki}+t_{kij}=0.
$$
In this space there are two irreducible representations of the rotations group, and any splitting of this 10-dimensional space (invariant under the action of the rotations group) into two 5-dimensional spaces gives two irreducible representations of the rotations group. In \cite{Gelfand:2018}, a method is proposed for splitting a 10-dimensional space into two 5-dimensional ones using additional conditions imposed on a tensor components. Condition
\begin{equation}
t_{ijk}=q\;t_{jki}=q^2\;t_{kij}
\label{intr:cyclic with q}
\end{equation}
determines one 5-dimensional subspace and condition
\begin{equation}
t_{ijk}=q^2\;t_{jki}=q\;t_{kij}
\label{intr:cyclic with q^2}
\end{equation}
the second. First, it is worth to mention that these conditions are invariant under the representation (\ref{intr:representation of SO(3)}) of the rotations group $\mbox{SO}(3)$. Second, in each 5-dimensional subspace we have irreducible unitary representation of $\mbox{SO}(3)$.

In \cite{Gelfand:2018}, the authors work with tensor components. In the present paper, we work with generators $\vartheta^1,\vartheta^2,\vartheta^3$ of tensor algebra. We find a set of ternary relations for the generators $\vartheta^1,\vartheta^2,\vartheta^3$ of tensor algebras that determine a 10-dimensional subspace of representations of weight 2. An analysis of these ternary relations shows that part of them can be obtained from binary relations. Thus, the study of tensor representations of weight 2 of the rotations group leads us to a unital associative algebra whose generators satisfy the relations
\begin{eqnarray}
&&\delta_{ij}\vartheta^i\vartheta^j=0,\qquad\qquad\qquad\qquad\;\;\;\;\mbox{(binary relations)},\nonumber\\
&&\vartheta^i\vartheta^j\vartheta^k+\vartheta^j\vartheta^k\vartheta^i+\vartheta^k\vartheta^i\vartheta^j=0,\;\; \mbox{(ternary relations)},\nonumber
\end{eqnarray}
which we call a ternary algebra associated with representations of weight 2 of the rotations group.
In analogy with (\ref{intr:ternary relations with q}) and (\ref{intr:cyclic with q^2}), we split this algebra into two algebras. The first algebra is generated by $\theta^1,\theta^2,\theta^3$, which are subjected to the relations
\begin{eqnarray}
&&\delta_{ij}\theta^i\theta^j=0,\;\;\;\;\;\;\qquad\qquad\qquad\qquad\;\;\; \mbox{(binary relations)}\nonumber\\
&&\theta^{i}\theta^{j}\theta^{k}
              =q\;\theta^{j}\vartheta^{k}\theta^{i},\qquad\qquad\qquad\quad\;\;\mbox{(ternary relations)},\nonumber
\end{eqnarray}
and we call this algebra $q$-algebra associated with irreducible representation of weight 2 of the rotations group. The second algebra is generated by $\bar\theta^1,\bar\theta^2,\bar\theta^3$, which are subjected to the set of relations
\begin{eqnarray}
&&\delta_{ij}\bar\theta^i\bar\theta^j=0,\;\;\;\;\;\;\qquad\qquad\qquad\qquad\;\;\; \mbox{(binary relations)}\nonumber\\
&&\bar\theta^{i}\bar\theta^{j}\bar\theta^{k}
              =\bar q\;\bar\theta^{j}\bar\vartheta^{k}\bar\theta^{i},\qquad\qquad\qquad\quad\;\;\mbox{(ternary relations)},\nonumber
\end{eqnarray}
and we call this algebra $\bar q$-algebra associated with irreducible representation of weight 2 of the rotations group.

We show that the subspace of $q$-algebra spanned by triple products $\theta^i\theta^j\theta^k$ is 5-dimensional and construct a basis for this subspace. Recall that in this 5-dimensional subspace we have an irreducible representation of $\mbox{SO}(3)$. We define Hermitian scalar product in this 5-dimensional subspace and find the explicit formula for $5\times 5$-matrix corresponding to infinitesimal rotation. This matrix is skew-Hermitian and traceless, thus we have the representation of Lie algebra $\mbox{so}(3)$ in $\mbox{so}(5)$. The structure of this matrix surprisingly coincides with the structure of representation $\bf 10$ of $\mbox{SU}(5)$ proposed by H. Georgi and S. Glashow in their GUT model for left-handed fermions of a single generation \cite{Croon_others:Front_Phys_2019},\cite{Georgi:Perseus_Books_1999}, \cite{Glashow:Nucl_Phys_1961}. We discuss this in the last section (Discussion).

In conclusion, we draw attention to binary relations $\delta_{ij}\theta^i\theta^j=0$ in all three algebras listed above. These relations arise naturally from the requirement of irreducible representations. On the other hand, these relations fit perfectly into the quark model and can be explained in terms of the quark model as follows. Let $R,G,B$ be color charges and $\bar R,\bar G,\bar B$ be anti-colors. Carrier particles of strong interactions (gluons) carry a color charge consisting of a pair, color and anti-color, where the colors themselves are different. Thus, the interaction changes the color charge of a quark. However, the colorless combination $R\bar R+B\bar B+G\bar G$ in no way interacts with quarks and therefore should vanish, i.e. we must have $R\bar R+B\bar B+G\bar G=0$. We think that this is the physical source of the binary relations of $q$-algebra and $\bar q$-algebra.
\section{Tensor representations of $\mbox{SO}(3)$}
We will consider irreducible representations of the group of rotations in three-dimensional Euclidean space $E$. The Euclidean metric of $E$ will be denoted by $<\;,\;>:E\times E\to \mathbb R$. We fix an orthonormal basis $e_1,e_2,e_3$ for $E$. Hence $<e_i,e_j>=\delta_{ij}$. For any vector $x\in E$ we have $x=x^i\,e_i$, where $x^1,x^2,x^3$ are coordinates of $x$. The rotations group $\mbox{SO}(3)$ of $E$ is the group of special orthogonal matrices $g=(g^i_j)$, where $g$ satisfies $g\,g^t=g^t\,g=I$ ($I$ is the third order unit matrix) and $\mbox{Det}\,g=1.$

Let $E^\ast$ be the dual space for $E$ and $\vartheta^1,\vartheta^2,\vartheta^3$ be the dual basis for $e_1,e_2,e_3$, i.e. $\vartheta^i(e_j)=\delta^i_j$. We will consider vector spaces of complex-valued covariant tensors in three-dimensional Euclidean space $E$, which transform by means of the rotations group, when we pass in $E$ from one Cartesian coordinate system to another. Let us denote the tensor algebra of complex-valued covariant tensors transforming with the help of $\mbox{SO}(3)$ by $T(E)$. This algebra is unital associative algebra over $\mathbb C$ and it is the direct sum of subspaces
\be
T(E)=\mathbb C\oplus T^1(E)\oplus T^2(E)\oplus T^3(E)\oplus\ldots,
\ee
where $T^a(E)$ is a subspace of tensors of rank $a$. The product of tensors $t=(t_{i_1i_2\ldots i_a}),s=(s_{i_1i_2\ldots i_a})$ of ranks $a$ and $b$ respectively is the tensor of rank $a+b$, whose components are
\be
(t\cdot s)_{i_1\ldots i_ai_{a+1}\ldots i_{a+b}}=
   t_{i_1\ldots i_a}\,s_{i_{a+1}\ldots i_{a+b}}.
\label{form:multiplication of tensors}
\ee
In a Cartesian coordinate system, determined by an orthonormal (oriented) basis  $e_1,e_2,e_3$, we can consider $\vartheta^1,\vartheta^2,\vartheta^3$ as generators of tensor algebra $T(E)$ and associate to each complex-valued covariant tensor of $a$th rank $t=(t_{i_1i_2\ldots i_a})$ the $a$th order polynomial $t(\vartheta)$, where
\be
t(\vartheta)=t_{i_1i_2\ldots i_a}\,\vartheta^{i_1}\otimes \vartheta^{i_2}\otimes\ldots\otimes\vartheta^{i_a}.
\label{form:invariantelement}
\ee
The multiplication of tensors (\ref{form:multiplication of tensors}) can be extended to polynomials (\ref{form:invariantelement}) as follows: the product of two homogeneous polynomials
\be
t(\vartheta)=t_{i_1i_2\ldots i_a}\,\vartheta^{i_1}\otimes \vartheta^{i_2}\otimes\ldots\otimes\vartheta^{i_a},\;\;
   s(\vartheta)=s_{j_1j_2\ldots j_b}\,\vartheta^{j_1}\otimes \vartheta^{j_2}\otimes\ldots\otimes\vartheta^{j_b},
\ee
is the polynomial
\be
t(\vartheta)\cdot s(\vartheta)=t_{i_1i_2\ldots i_a}s_{j_1j_2\ldots j_b}
       \,\vartheta^{i_1}\otimes \vartheta^{i_2}\otimes\ldots\otimes\vartheta^{i_a}\otimes
                \vartheta^{j_1}\otimes \vartheta^{j_2}\otimes\ldots\otimes\vartheta^{j_b}.
\label{form:multiplication theta expressions}
\ee
In the present paper, it will be convenient for us to distinguish between the algebra of covariant tensors $T(E)$ and the algebra of polynomials (\ref{form:invariantelement}), which will be denoted $T_\vartheta(E)$ and its subspace of homogeneous polynomials of $a$th order will be denoted by $T^a_\vartheta(E)$.

All possible monomials $\{\vartheta^{i_1}\otimes \vartheta^{i_2}\otimes\ldots\vartheta^{i_a}\}$ (in some way ordered), where $a=0,1,2,\ldots$ and $a=0$ is identified with the identity element $\mathbf{1}$ of the tensor algebra $T_\vartheta(E)$, form a basis for the vector space $T_\vartheta(E)$. Obviously the multiplication (\ref{form:multiplication theta expressions}) in terms of this basis is given by
\begin{eqnarray}
&&(\vartheta^{i_1}\otimes \vartheta^{i_2}\otimes\ldots\otimes\vartheta^{i_a})\cdot
                (\vartheta^{j_1}\otimes \vartheta^{j_2}\otimes\ldots\otimes\vartheta^{j_b})=\nonumber\\
                     &&\qquad\qquad\qquad\qquad\vartheta^{i_1}\otimes \vartheta^{i_2}\otimes\ldots\vartheta^{i_a}\otimes
                                     \vartheta^{j_1}\otimes \vartheta^{j_2}\otimes\ldots\otimes\vartheta^{j_b}.\nonumber
\end{eqnarray}
When we rotate the three-dimensional space $E$ by means of a special orthogonal matrix $g=(g^i_j)$, i.e. $e_i\to g^j_i\,e_j$, then components $t_{i_1i_2\ldots i_a}$ of an $a$th rank tensor $t$ transform according to the formula
\be
t_{i_1i_2\ldots i_a}\to g_{i_1}^{j_1}g_{i_2}^{j_2}\ldots g_{i_a}^{j_a}\;t_{j_1j_2\ldots j_a},
\label{form:transformation of components of tensor}
\ee
which determines a linear transformation in a vector space of rank $a$ tensors and, hence, a representation of the rotations group $\mbox{SO}(3)$ in $T^a(E)$. Monomials of the basis
$\vartheta^{i_1}\otimes \vartheta^{i_2}\otimes\ldots\vartheta^{i_a}$ transform under rotation $e_i\to g^j_i\,e_j$ according to the formula
\be
\vartheta^{i_1}\otimes \vartheta^{i_2}\otimes\ldots\otimes\vartheta^{i_a}\to
     \tilde{g}_{j_1}^{i_1}\tilde{g}_{j_2}^{i_2}\ldots \tilde{g}_{j_a}^{i_a}\;
          \vartheta^{j_1}\otimes \vartheta^{j_2}\otimes\ldots\otimes\vartheta^{j_a}
\label{form:transformation of monomials}
\ee
where $\tilde g=(\tilde g^i_j)$ is reciprocal matrix of $g$. If we now transform both tensor components and basic monomials in expression (\ref{form:invariantelement}), then the whole expression (\ref{form:invariantelement}) will not change, which expresses an invariance of elements of tensor algebra with respect to linear transformations of $E$.

However, in this paper we need to define an action of a linear operator induced by a special orthogonal matrix $g$ on the elements of tensor algebra $t(\vartheta)$. We denote a linear operator in a vector space $T^a_\vartheta(E)$ corresponding to $g\in\mbox{SO}(3)$ by $R_g$ and define its action on elements of tensor algebra as follows
\be
R_g\, \big(t(\vartheta)\big)=t_{i_1i_2\ldots i_a}\,R_g(\vartheta^{i_1}\otimes \vartheta^{i_2}\otimes\ldots\otimes\vartheta^{i_a}),
\ee
where
\be
R_g(\vartheta^{i_1}\otimes \vartheta^{i_2}\otimes\ldots\vartheta^{i_a})=
  {g}_{j_1}^{i_1}{g}_{j_2}^{i_2}\ldots {g}_{j_a}^{i_a}\;
          \vartheta^{j_1}\otimes \vartheta^{j_2}\otimes\ldots\otimes\vartheta^{j_a}.
\label{form:representation of SO(3) on varthetas}
\ee
The main construction that we will consider in this paper can be described in general terms as follows. Suppose $S$  is a linear operator constructed by means of operator $R_g$ and related with irreducible representations of $\mbox{SO}(3)$. Equation
\be
S\big(t(\vartheta)\big)=0,
\label{form:general equation}
\ee
defines the subspace in a vector space $T^a_\vartheta(E)$. This subspace can be described in two equivalent ways. The first one is that we consider all basic monomials to be independent, that is, we assume that there are no algebraic relations between them, and then we collect all terms in the left-hand side of (\ref{form:general equation}) containing the same basic monomial and put its coefficient, which is a linear combination of tensor components, equal to zero. Thus, we obtain a system of linear equations for the components of a tensor and the set of solutions of this system of equations will determine a subspace in $T^a(E)$.

An alternative way to describe the same subspace is to assume that all components of a tensor are independent, collect the terms on the left-hand side of equation (\ref{form:general equation}) containing the same component of a tensor, and set the coefficient of this component to zero. Since this coefficient will be a linear combination of monomials, we obtain algebraic relations between the $a$th order products of generators. The next problem that we solve in this paper is the selection of a system of independent relations. In addition, we consider an interesting question whether the obtained independent algebraic relations of $a$th order (or a part of them) can be obtained from relations of order $a-1$.

It is well known that when considering representations of the rotations group in three-dimensional space, it suffices to restrict ourselves to unitary representations, because one can always introduce a scalar product in the representation space with the help of invariant integration over the rotations group in such a way that the representation becomes unitary. In the case of a vector space of tensors of rank $a$ a Hermitian scalar product can be defined as follows
\be
<t(\vartheta),s(\vartheta)>=
                  \sum_{i_1i_2\ldots i_a}t_{i_1i_2\ldots i_a}\,\bar{s}_{i_1i_2\ldots i_a}.
\ee
Hence the basic monomials $\vartheta^{i_1}\otimes \vartheta^{i_2}\otimes\ldots\otimes\vartheta^{i_a}$ are orthonormal, i.e.
\be
<\vartheta^{i_1}\otimes \vartheta^{i_2}\otimes\ldots\otimes\vartheta^{i_a},
                                 \vartheta^{j_1}\otimes \vartheta^{j_2}\otimes\ldots\otimes\vartheta^{j_a}>=
    \delta^{i_1j_1}\,\delta^{i_2j_2}\ldots\,\delta^{i_aj_a}.
\ee
It is easy to show that a linear transformation $R_g$ for any $g\in \mbox{SO}(3)$ is now unitary. For a special orthogonal matrix $g=(g^i_j)$ we have
$$
g^i_k\,\delta^{kl}\,g^j_l=\delta^{ij}.
$$
Thus
\begin{eqnarray}
&& <R_g(\vartheta^{i_1}\otimes\ldots\otimes\vartheta^{i_a}),R_g(
                     \vartheta^{j_1}\otimes\ldots\otimes\vartheta^{j_a}>\nonumber\\
&& \qquad=g^{i_1}_{k_1}\ldots g^{i_a}_{k_a}g^{j_1}_{l_1}\ldots g^{j_a}_{l_a}
         <\vartheta^{k_1}\otimes\ldots\otimes\vartheta^{k_a},\vartheta^{l_1}\otimes\ldots\otimes\vartheta^{l_a}>\nonumber\\
&&\qquad =g^{i_1}_{k_1}\delta^{k_1l_1}g^{j_1}_{l_1}\ldots g^{i_a}_{k_a}\delta^{k_al_a}g^{j_a}_{l_a}
          =\delta^{i_1j_1}\ldots \delta^{i_aj_a}\nonumber\\
&&\qquad=<\vartheta^{i_1}\otimes\ldots\otimes\vartheta^{i_a},\vartheta^{j_1}\otimes\ldots\otimes\vartheta^{j_a}>.\nonumber
\end{eqnarray}
From this it follows that $R:\mbox{SO}(3)\to \mbox{U}(T^a_\vartheta(E))$, where $\mbox{U}(T^a_\vartheta(E))$ is the group of unitary linear transformations of a vector space $T^a_\vartheta(E)$.

In conclusion of the first section, we simplify the notations that we will use in this paper. The symbol of the tensor product $\otimes$ requires a lot of space in a text. Therefore, in what follows we will omit it and write $\vartheta^{i_1}\vartheta^{i_2}\ldots\vartheta^{i_a}$ instead of $\vartheta^{i_1}\otimes\vartheta^{i_2}\otimes\ldots\otimes\vartheta^{i_a}$.
\section{Algebra associated with tensor representations of weight 2 of $\mbox{SO}(3)$}
In this section, we will study algebraic relations between generators $\vartheta^1,\vartheta^2,\vartheta^3$ of tensor algebra $T_\vartheta(E)$, which arise in relation with irreducible representations of weight 2 of the rotations group in the space of rank 3 tensors $T^3_\vartheta(E)$.

Recall that the Lie algebra $\mbox{so}(3)$ of the rotations group $\mbox{SO}(3)$ is the vector space of skew-symmetric matrices of 3rd order equipped with the Lie bracket, which is the usual commutator of two matrices. The matrices
\be
A_1=\left(
      \begin{array}{ccc}
        0 & 0 & 0 \\
        0 & 0 & -1 \\
        0 & 1 & 0 \\
      \end{array}
    \right),
A_2=\left(
      \begin{array}{ccc}
        0 & 0 & 1 \\
        0 & 0 & 0 \\
        -1 & 0 & 0 \\
      \end{array}
    \right),
A_3=\left(
      \begin{array}{ccc}
        0 & -1 & 0 \\
        1 & 0 & 0 \\
        0 & 0 & 0 \\
      \end{array}
    \right).\nonumber
\ee
can be taken as a basis for $\mbox{so}(3)$.
Then $[A_i,A_j]=\epsilon_{ijk}\,A_k$, where $\epsilon_{ijk}$ is the Levi-Civita tensor. The matrices $A_1$,$A_2$,$A_3$ generate infinitesimal rotations around coordinate axes $x^1,x^2,x^3$ respectively.

A representation of the rotations group $\mbox{SO}(3)$ in the vector space of rank 3 tensors $T^3_\vartheta(E)$, defined by the formula (\ref{form:representation of SO(3) on varthetas}), is denoted by $R:\mbox{SO}(3)\to \mbox{U}(T^3_\vartheta(E))$. At the infinitesimal level, we have a representation of the Lie algebra $\mbox{so}(3)$. The corresponding mapping, which assigns to each matrix of $\mbox{so}(3)$ the linear transformation in the vector space of rank 3 tensors $T^3_\vartheta(E)$, will be denoted by $dR$, and for any matrix $F\in \mbox{so}(3)$ its image (linear transformation in $T^3_\vartheta(E)$) will be denoted by $dR_F$. Obviously $dR$ is tangent mapping for $R$ at the unit of the rotations group $\mbox{SO}(3)$.

In \cite{Gelfand:2018} it is proved that all irreducible tensor representations of the rotations group $\mbox{SO}(3)$ in the vector space of rank 3 tensors can be described by means of the linear operator
\be
H^2=-(dR_{A_1})^2-(dR_{A_2})^2-(dR_{A_3})^2.
\ee
More precisely, this can be formulated as follows: the solutions of the equation
\be
H^2\,t(\vartheta)=l(l+1)\,t(\vartheta),
\label{form:basic equation}
\ee
where
\be
t(\vartheta)=t_{ijk}\;\vartheta^i\,\vartheta^j\,\vartheta^k,
\ee
determine the subspace of rank 3 tensors, where we have a multiple irreducible representation of weight $l$ of the rotations group , i.e. this representation is either irreducible or can be split into irreducible representations. In particular, we note that the representation of weight $l=0$ is irreducible and the representation space is the one-dimensional space of skew-symmetric tensors of rank 3. Here we see a connection with the Grassmann algebra and the theory of differential forms (skew-symmetric tensor fields) in three-dimensional space.

In this paper, we are interested in the representations of weight $l=2$. In \cite{Gelfand:2018}, this representation is described and it is shown there that the representation space is 10-dimensional and in this space we have a double irreducible representation of $\mbox{SO}(3)$, i.e. one can decompose it into two irreducible representations by splitting the 10-dimensional representation space into two 5-dimensional subspaces invariant under representation of $\mbox{SO}(3)$. We will get the same results, but in contrast to the approach of \cite{Gelfand:2018}, where the authors use components of tensors and the operator $H^2$ in the equation (\ref{form:basic equation}) acts on components $t_{ijk}$ by means of infinitesimal version of (\ref{form:transformation of components of tensor}), we will use the generators $\vartheta^1,\vartheta^2,\vartheta^3$ of tensor algebra and the operator $H^2$ will be applied to triple products $\vartheta^i\,\vartheta^j\,\vartheta^k$ and this will lead us to an interesting algebraic structure.

First of all, we find a linear transformation $dR_F$ of the vector space of rank 3 tensors for any skew-symmetric matrix $F=(F^i_j)\in \mbox{so}(3)$. Each skew-symmetric matrix $F$ generates the one-parameter subgroup $g_\tau=\exp(\tau\,F)$ of the rotations group $\mbox{SO}(3)$.
Expanding $g_\tau$ in a power series in $\tau$ in (\ref{form:representation of SO(3) on varthetas}) and selecting terms linear in $\tau$, we obtain
\begin{eqnarray}
R_{g_\tau}(\vartheta^i\,\vartheta^j\,\vartheta^k) &=& (\delta^i_m+F^i_m\,\tau+\ldots)(\delta^j_n+F^i_n\,\tau+\ldots)
          (\delta^k_p+F^k_p\,\tau+\ldots)\nonumber\\
          &&\qquad\qquad\qquad\qquad\qquad\qquad\qquad\qquad\qquad\quad \times\vartheta^m\vartheta^n\vartheta^p\nonumber\\
    &=& \vartheta^i\,\vartheta^j\,\vartheta^k+\big(F^i_p\vartheta^p\vartheta^j\vartheta^k
                                                          +F^j_p\vartheta^i\vartheta^p\vartheta^k+
                       F^k_p\vartheta^i\vartheta^j\vartheta^p\big)\,\tau+\ldots.\nonumber
\end{eqnarray}
Hence, the linear transformation $dR_F$ in the vector space of rank 3 tensors, generated by a matrix $F\in\mbox{so}(3)$, can be written as follows
\be
dR_F(\vartheta^i\,\vartheta^j\,\vartheta^k)=
       F^i_p\vartheta^p\vartheta^j\vartheta^k+F^j_p\vartheta^i\vartheta^p\vartheta^k+
                                                      F^k_p\vartheta^i\vartheta^j\vartheta^p,
\ee
or
\be
dR_F(t(\vartheta))=t_{ijk}(F^i_p\vartheta^p\vartheta^j\vartheta^k+F^j_p\vartheta^i\vartheta^p\vartheta^k+
                       F^k_p\vartheta^i\vartheta^j\vartheta^p).
\label{form:formula for dR_F}
\ee
Since the linear operator at the left-hand side of the basic equation (\ref{form:basic equation}) is the sum of squares of linear operators, generated by infinitesimal rotations $A_1,A_2,A_3$, we apply the found formula twice and get
\begin{eqnarray}
&&(dR_F)^2\,t(\vartheta)=t_{ijk}\big((F^2)^i_p\,\vartheta^p\vartheta^j\vartheta^k+
                                    (F^2)^j_p\,\vartheta^i\vartheta^p\vartheta^k+
                                       (F^2)^k_p\,\vartheta^i\vartheta^j\vartheta^p\label{form:square of dR_F}\\
         &&\qquad\qquad\qquad\qquad+2\,F^i_pF^j_r\vartheta^p\vartheta^r\vartheta^k+
            2\,F^i_pF^k_r\vartheta^p\vartheta^j\vartheta^r+
                2\,F^j_pF^k_r\vartheta^i\vartheta^p\vartheta^r\big).\nonumber
\end{eqnarray}
Now we can calculate the operator $H^2=-(dR_{A_1})^2-(dR_{A_2})^2-(dR_{A_3})^2$ successively substituting $A_1,A_2,A_3$ into place of $F$. We obtain
$$
-(dR_{A_1})^2t(\vartheta)=\mathfrak R^\prime_1+\mathfrak R^{\prime\prime}_1-\mathfrak R^{\prime\prime\prime}_1,
$$
where
\begin{eqnarray}
\mathfrak R^\prime_1 \!\!\!\!\!&=&\!\!\!\!\! (t_{2jk}\vartheta^2+t_{3jk}\vartheta^3)\vartheta^j\vartheta^k+
                            \vartheta^j(t_{j2k}\vartheta^2+t_{j3k}\vartheta^3)\vartheta^k+
                              \vartheta^j\vartheta^k(t_{jk2}\vartheta^2+t_{jk3}\vartheta^3),\nonumber\\
\mathfrak R^{\prime\prime}_1 \!\!\!\!\!&=&\!\!\!\!\! (t_{23k}\vartheta^3\vartheta^2+t_{32k}\vartheta^2\vartheta^3)\vartheta^k+
   t_{2k3}\vartheta^3\vartheta^k\vartheta^2+t_{3k2}\vartheta^2\vartheta^k\vartheta^3\nonumber\\
   &&\qquad\qquad\qquad\qquad\qquad\qquad\qquad\qquad\qquad
       +\vartheta^k(t_{k23}\vartheta^3\vartheta^2+t_{k32}\vartheta^2\vartheta^3)\nonumber\\
\mathfrak R^{\prime\prime\prime}_1\!\!\!\!\!&=&\!\!\!\!\!
   \big(t_{22k}(\vartheta^3)^2+t_{33k}(\vartheta^2)^2\big)\vartheta^k+
     t_{2j2}\vartheta^3\vartheta^j\vartheta^3+t_{3j3}\vartheta^2\vartheta^j\vartheta^2\nonumber\\
     &&\qquad\qquad\qquad\qquad\qquad\qquad\qquad\qquad\qquad
       +\vartheta^k\big(t_{k33}(\vartheta^2)^2+t_{k22}(\vartheta^3)^2\big).\nonumber
\end{eqnarray}
Similarly we obtain
\be
-(dR_{A_2})^2t(\vartheta)=\mathfrak R^\prime_2+\mathfrak R^{\prime\prime}_2-\mathfrak R^{\prime\prime\prime}_2,\;\;
  -(dR_{A_3})^2t(\vartheta)=\mathfrak R^\prime_3+\mathfrak R^{\prime\prime}_3-\mathfrak R^{\prime\prime\prime}_3,
\ee
where
\begin{eqnarray}
\mathfrak R^\prime_2 \!\!\!\!\!&=&\!\!\!\!\! (t_{3jk}\vartheta^3+t_{1jk}\vartheta^1)\vartheta^j\vartheta^k+
                            \vartheta^j(t_{j3k}\vartheta^3+t_{j1k}\vartheta^1)\vartheta^k+
                              \vartheta^j\vartheta^k(t_{jk1}\vartheta^1+t_{jk3}\vartheta^3),\nonumber\\
\mathfrak R^{\prime\prime}_2 \!\!\!\!\!&=&\!\!\!\!\! (t_{31k}\vartheta^1\vartheta^3+t_{13k}\vartheta^3\vartheta^1)\vartheta^k+
   t_{3k1}\vartheta^1\vartheta^k\vartheta^3+t_{1k3}\vartheta^3\vartheta^k\vartheta^1\nonumber\\
   &&\qquad\qquad\qquad\qquad\qquad\qquad\qquad\qquad\qquad
       +\vartheta^k(t_{k31}\vartheta^1\vartheta^3+t_{k13}\vartheta^3\vartheta^1)\nonumber\\
\mathfrak R^{\prime\prime\prime}_2\!\!\!\!\!&=&\!\!\!\!\!
   \big(t_{33k}(\vartheta^1)^2+t_{11k}(\vartheta^3)^2\big)\vartheta^k+
     t_{3j3}\vartheta^1\vartheta^j\vartheta^1+t_{1j1}\vartheta^3\vartheta^j\vartheta^3\nonumber\\
     &&\qquad\qquad\qquad\qquad\qquad\qquad\qquad\qquad\qquad
       +\vartheta^k\big(t_{k11}(\vartheta^3)^2+t_{k33}(\vartheta^1)^2\big),\nonumber
\end{eqnarray}
and
\begin{eqnarray}
\mathfrak R^\prime_3 \!\!\!\!\!&=&\!\!\!\!\! (t_{1jk}\vartheta^1+t_{2jk}\vartheta^2)\vartheta^j\vartheta^k+
                            \vartheta^j(t_{j1k}\vartheta^1+t_{j2k}\vartheta^2)\vartheta^k+
                              \vartheta^j\vartheta^k(t_{jk1}\vartheta^1+t_{jk2}\vartheta^2),\nonumber\\
\mathfrak R^{\prime\prime}_3 \!\!\!\!\!&=&\!\!\!\!\! (t_{12k}\vartheta^2\vartheta^1+t_{21k}\vartheta^1\vartheta^2)\vartheta^k+
   t_{1j2}\vartheta^2\vartheta^j\vartheta^1+t_{2j1}\vartheta^1\vartheta^j\vartheta^2\nonumber\\
   &&\qquad\qquad\qquad\qquad\qquad\qquad\qquad\qquad\qquad
       +\vartheta^k(t_{k12}\vartheta^2\vartheta^1+t_{k21}\vartheta^1\vartheta^2)\nonumber\\
\mathfrak R^{\prime\prime\prime}_3\!\!\!\!\!&=&\!\!\!\!\!
   \big(t_{11k}(\vartheta^2)^2+t_{22k}(\vartheta^1)^2\big)\vartheta^k+
     t_{1j1}\vartheta^2\vartheta^j\vartheta^2+t_{2j2}\vartheta^1\vartheta^j\vartheta^1\nonumber\\
     &&\qquad\qquad\qquad\qquad\qquad\qquad\qquad\qquad\qquad
       +\vartheta^k\big(t_{k22}(\vartheta^1)^2+t_{k11}(\vartheta^2)^2\big).\nonumber
\end{eqnarray}
Since we are considering the tensor representations of weight 2, in the basic equation (\ref{form:basic equation}) $l$ is 2 and in this case the equation takes on the form
\be
H^2\,t(\vartheta)=6\,t(\vartheta).
\label{form:particular case of basic equation l=2}
\ee
It should be noted here that if in this equation we consider a tensor $t$ rather than an element of tensor algebra $t(\vartheta)$, then the solution of the equation is a subspace of rank 3 tensors that satisfy the conditions
\be
t_{ijk}+t_{jki}+t_{kij}=0,\;\;\sum_{i=1}^3\;t_{iij}=0,\;\sum_{i=1}^3\;t_{iji}=0,\;
                   \sum_{i=1}^3\;t_{jii}=0,
\label{form:subspace of traceless tensors}
\ee
where $j=1,2,3$, i.e. the trace of tensor $t$ for any pair of subscripts is zero. The subspace of rank 3 tensors, which is determined by (\ref{form:subspace of traceless tensors}), will be denoted by $\frak T^3(E)$. Thus $\frak T^3(E)\subset T^3(E)$.

Let us now see what kind of algebraic structure this equation will lead to if we substitute the corresponding element of tensor algebra instead of the tensor.
On the left-hand side of this equation, the terms denoted by $\mathfrak R^\prime_1,\mathfrak R^\prime_2,\mathfrak R^\prime_3$ add up to $2\,t(\vartheta)+2\,t(\vartheta)+2\,t(\vartheta)=6\,t(\vartheta)$. Thus, the basic equation (\ref{form:particular case of basic equation l=2}) takes on the form
\be
\sum_{i=1}^3\;(\mathfrak R^{\prime\prime}_i-\mathfrak R^{\prime\prime\prime}_i)=0.
\label{form:equation for varthetas}
\ee
Equation (\ref{form:equation for varthetas}) can be viewed from two different points of view. The first is that we consider all products $\vartheta^i\vartheta^j\vartheta^k$ to be independent, that is, there are no relations between the generators and their products. Then we collect all the terms that contain a triple product of the generators $\vartheta^i\,\vartheta^j\,\vartheta^k$. Taking this product out of brackets, we get a linear combination of components of tensor $t$ and, equating it to zero, we obtain the linear equation. Having done this for all products $\vartheta^i\,\vartheta^j\,\vartheta^k$, we obviously get a system of linear equations (\ref{form:subspace of traceless tensors}), which determines the subspace of rank 3 tensors with double irreducible representation of the rotations group.

Another point of view on equation (\ref{form:equation for varthetas}) is that we consider the components of a tensor $t$ as independent. Then, collecting all terms that contain some tensor component $t_{ijk}$ and taking this component out of brackets, we get a linear combination of triple products of generators $\vartheta^i\vartheta^j\vartheta^k$. Equating all combinations obtained in this way to zero, we obtain a set of algebraic relations that define some subalgebra of the tensor algebra $T_\vartheta(E)$. In this paper, we stand on the second point of view on equation (\ref{form:equation for varthetas}).

Let $S_3$ be the group of permutations of integers $(1,2,3)$. We will use sums of the form $\sum_{S_3}$, where under the sign of the sum there will be quantities depending on three indices $i,j,k$. This means that the sum is taken over all permutations of integers 1,2,3, that is, the triple $(i,j,k)$ runs through all six permutations. Now collecting terms that contain one and the same tensor component, we can write the left-hand side of the equation (\ref{form:equation for varthetas}) in the following form
\be
\sum_{i=1}^3\;(\mathfrak R^{\prime\prime}_i-\mathfrak R^{\prime\prime\prime}_i)=\mathfrak S_1+\mathfrak S_2+\mathfrak S_3,
\ee
where
\begin{eqnarray}
\mathfrak S_1 \!\!\!\!\!&=&\!\!\!\!\! \sum_{S_3}\;t_{ikj}(\vartheta^i\vartheta^j\vartheta^k+\vartheta^j\vartheta^k\vartheta^i+
     \vartheta^k\vartheta^i\vartheta^j)\nonumber\\
\mathfrak S_2 \!\!\!\!\!&=&\!\!\!\!\! \sum_{S_3}  \big(t_{jij}(\vartheta^i\vartheta^j\vartheta^j+\vartheta^j\vartheta^j\vartheta^i-(\vartheta^i)^3-\vartheta^k\vartheta^i\vartheta^k)\nonumber\\
 &&\qquad+t_{ijj}(\vartheta^j\vartheta^j\vartheta^i+\vartheta^j\vartheta^i\vartheta^j-(\vartheta^i)^3-\vartheta^i\vartheta^k\vartheta^k)\nonumber\\
   &&\qquad\qquad+t_{jji}(\vartheta^j\vartheta^i\vartheta^j+\vartheta^i\vartheta^j\vartheta^j-(\vartheta^i)^3-\vartheta^k\vartheta^k\vartheta^i)\big)\nonumber\\
\mathfrak S_3 \!\!\!\!\!&=&\!\!\!\!\! \frac{1}{2}\sum_{S_3}
  t_{iii}(\vartheta^j\vartheta^j\vartheta^i+\vartheta^j\vartheta^i\vartheta^j+\vartheta^i\vartheta^j\vartheta^j\nonumber\\
  &&\qquad\qquad\qquad\qquad \vartheta^k\vartheta^k\vartheta^i+\vartheta^k\vartheta^i\vartheta^k+\vartheta^i\vartheta^k\vartheta^k).\nonumber
\end{eqnarray}
The expression $\mathfrak S_1$ contains all terms with components $t_{ijk}$, where $i\neq j,j\neq k, i\neq k$. It is easy to see that this kind of terms we can find only in the expressions $\mathfrak R^{\prime\prime}_1,\mathfrak R^{\prime\prime}_2,\mathfrak R^{\prime\prime}_3$, when we put $k=1$ in $\mathfrak R^{\prime\prime}_1$, $k=2$ in $\mathfrak R^{\prime\prime}_2$ and $k=3$ in $\mathfrak R^{\prime\prime}_3$. Collecting all these terms, we get
the expression $\mathfrak S_1$. Hence the six components $t_{ijk}$ of tensor $t$, where $(i,j,k)$ is a permutation of $1,2,3$, give us six relations
\be
\vartheta^i\vartheta^k\vartheta^j+\vartheta^k\vartheta^j\vartheta^i+\vartheta^j\vartheta^i\vartheta^k=0,\;\;
     i\neq j,\;j\neq k,\;i\neq k.
\label{form:cyclic relations i,j,k}
\ee
To make formulas more compact, we introduce a notation for the sum of ternary products of generators obtained by cyclic permutation of generators. Let us denote
\be
\{\vartheta^i,\vartheta^k,\vartheta^j\}=
   \vartheta^i\vartheta^k\vartheta^j+\vartheta^k\vartheta^j\vartheta^i+\vartheta^j\vartheta^i\vartheta^k.
\label{form:definition of ternary curl brackets}
\ee
Now we can write the six relations (\ref{form:cyclic relations i,j,k}) by means of (\ref{form:definition of ternary curl brackets}) in a compact form as follows
\be
\{\vartheta^i,\vartheta^k,\vartheta^j\}=0,\;\;i\neq j,\;j\neq k,\;i\neq k.
\ee
Thus, we can formulate the first set of algebraic relations as follows: for three different generators, the sum of the products of all three, where each product is obtained from the previous one by cyclic permutation of the factors, is zero.

In order to write compactly the following relations, which we obtain from expression $\mathfrak S_2$, we denote $\vartheta^{ijk}=\vartheta^i\vartheta^j\vartheta^k$. The expression $\mathfrak S_2$ gives 18 relations for triple products of generators, which can be split into following six sets of relations
\begin{eqnarray}
&& \;\;\mbox{\bf I}\;(i=1,j=2,k=3)\qquad\qquad\quad \mbox{\bf IV}\;(i=2,j=1,k=3)\nonumber\\
&&\!\!\!\!\!\!\!\!\vartheta^{122}+\vartheta^{221}-\vartheta^{111}-\vartheta^{313}=0,\;\;\;\;\;
   \vartheta^{211}+\vartheta^{112}-\vartheta^{222}-\vartheta^{323}=0,\nonumber\\
&&\!\!\!\!\!\!\!\!\vartheta^{221}+\vartheta^{212}-\vartheta^{111}-\vartheta^{133}=0,\;\;\;\;\;
   \vartheta^{112}+\vartheta^{121}-\vartheta^{222}-\vartheta^{233}=0,\nonumber\\
&&\!\!\!\!\!\!\!\!\vartheta^{212}+\vartheta^{122}-\vartheta^{111}-\vartheta^{331}=0,\;\;\;\;\;
   \vartheta^{121}+\vartheta^{211}-\vartheta^{222}-\vartheta^{332}=0,\nonumber\\
&&\nonumber\\
&& \;\;\mbox{\bf II}\;(i=1,j=3,k=2)\qquad\qquad\quad \mbox{\bf V}\;(i=3,j=1,k=2)\nonumber\\
&&\!\!\!\!\!\!\!\!\vartheta^{133}+\vartheta^{331}-\vartheta^{111}-\vartheta^{212}=0,\;\;\;\;\;
   \vartheta^{311}+\vartheta^{113}-\vartheta^{333}-\vartheta^{232}=0,\nonumber\\
&&\!\!\!\!\!\!\!\!\vartheta^{331}+\vartheta^{313}-\vartheta^{111}-\vartheta^{122}=0,\;\;\;\;\;
   \vartheta^{113}+\vartheta^{131}-\vartheta^{333}-\vartheta^{322}=0,\nonumber\\
&&\!\!\!\!\!\!\!\!\vartheta^{313}+\vartheta^{133}-\vartheta^{111}-\vartheta^{221}=0,\;\;\;\;\;
   \vartheta^{131}+\vartheta^{311}-\vartheta^{333}-\vartheta^{223}=0,\nonumber\\
&&\nonumber\\
&& \;\;\mbox{\bf III}\;(i=2,j=3,k=1)\qquad\qquad\quad \mbox{\bf VI}\;(i=3,j=2,k=1)\nonumber\\
&&\!\!\!\!\!\!\!\!\vartheta^{233}+\vartheta^{332}-\vartheta^{222}-\vartheta^{121}=0,\;\;\;\;\;
   \vartheta^{322}+\vartheta^{223}-\vartheta^{333}-\vartheta^{131}=0,\nonumber\\
&&\!\!\!\!\!\!\!\!\vartheta^{332}+\vartheta^{323}-\vartheta^{222}-\vartheta^{211}=0,\;\;\;\;\;
   \vartheta^{223}+\vartheta^{232}-\vartheta^{333}-\vartheta^{311}=0,\nonumber\\
&&\!\!\!\!\!\!\!\!\vartheta^{323}+\vartheta^{233}-\vartheta^{222}-\vartheta^{112}=0,\;\;\;\;\;
   \vartheta^{232}+\vartheta^{322}-\vartheta^{333}-\vartheta^{113}=0,\nonumber
\end{eqnarray}
Finally, the expression $\mathfrak S_3$ gives three relations, which can be written as follows
\begin{eqnarray}
\{\vartheta^3,\vartheta^3,\vartheta^1\}+\{\vartheta^1,\vartheta^2,\vartheta^2\} &=& 0,\label{form:third group relations 1}\\
\{\vartheta^3,\vartheta^3,\vartheta^2\}+\{\vartheta^2,\vartheta^1,\vartheta^1\} &=& 0,\label{form:third group relations 2}\\
\{\vartheta^2,\vartheta^2,\vartheta^3\}+\{\vartheta^3,\vartheta^1,\vartheta^1\} &=& 0.\label{form:third group relations 3}
\end{eqnarray}
We obtained 27 relations and this result is consistent with our approach. Indeed, we consider all components of a tensor to be independent, in the case of a rank 3 tensor their number is 27 and, since each component gives one relation, we obtained all the relations.

Our next goal is to find a set of defining (independent) relations such that the above 27 relations would follow from them. First of all, we add up all relations in groups {\bf I} and {\bf II}. We get
$$
\{\vartheta^3,\vartheta^3,\vartheta^1\}+\{\vartheta^1,\vartheta^2,\vartheta^2\} = 6\,(\vartheta^1)^3.
$$
But, according to the relation (\ref{form:third group relations 1}), the left-hand side of the above relation is zero. Thus $(\vartheta^1)^3=0$. Similarly adding up all relations in groups {\bf III, IV} ({\bf V,VI}) and applying the relation (\ref{form:third group relations 2}) \big((\ref{form:third group relations 3})\big) we get $(\vartheta^2)^3=(\vartheta^3)^3=0$. Thus, we proved that from the set of relations {\bf I} - {\bf VI} and (\ref{form:third group relations 1}) - (\ref{form:third group relations 3}) it follows that the cube of each generator is 0, i.e.
$$
(\vartheta^1)^3=0,\;\;\;(\vartheta^2)^3=0,\;\;\;(\vartheta^3)^3=0.
$$
Next, we add up all relations of group {\bf I} and, taking into account that $(\vartheta^1)^3=0$, we obtain
$$
\{\vartheta^3,\vartheta^3,\vartheta^1\}= 2\,\{\vartheta^1,\vartheta^2,\vartheta^2\}.
$$
Substituting this result into (\ref{form:third group relations 1}), we obtain
$$
3\,\{\vartheta^1,\vartheta^2,\vartheta^2\}=0,
$$
and, consequently, $\{\vartheta^1,\vartheta^2,\vartheta^2\}=\{\vartheta^3,\vartheta^3,\vartheta^1\}=0$. Applying the same reasoning to the set of relations {\bf II} - {\bf VI}, we get analogous relations. Thus, summing up the still not final result of our considerations, we can collect together the obtained relations
\begin{eqnarray}
&&\{\vartheta^1,\vartheta^2,\vartheta^3\}=0,\;\;\;\{\vartheta^3,\vartheta^2,\vartheta^1\}=0,\label{form:cyclic relations 1}\\
&&\{\vartheta^1,\vartheta^1,\vartheta^2\}=0,\;\;\;\{\vartheta^2,\vartheta^2,\vartheta^1\}=0,\label{form:cyclic relations 2}\\
&&\{\vartheta^1,\vartheta^1,\vartheta^3\}=0,\;\;\;\{\vartheta^3,\vartheta^3,\vartheta^1\}=0,\label{form:cyclic relations 3}\\
&&\{\vartheta^2,\vartheta^2,\vartheta^3\}=0,\;\;\;\{\vartheta^3,\vartheta^3,\vartheta^2\}=0,\label{form:cyclic relations 4}\\
&&\;(\vartheta^1)^3=0,\;\;\;(\vartheta^2)^3=0,\;\;\;(\vartheta^3)^3=0.\label{form:cyclic relations 5}
\end{eqnarray}
It is worth to mention that all these relations are independent, i.e. none of them can be obtained as a consequence of the rest. It should also be noted that all these relations can be summarized as follows: the sum of triple products of generators, obtained by cyclic permutations of factors from any triple product of generators $\vartheta^1,\vartheta^2,\vartheta^3$, is zero, i.e.
\be
\{\vartheta^i,\vartheta^j,\vartheta^k\}=0,
\ee
for any triple $i,j,k$ of integers 1,2,3.

However, the set of relations (\ref{form:cyclic relations 1}) - (\ref{form:cyclic relations 5}) is not equivalent to the initial set of 27 relations. It is easy to see that one can not get any relation in {\bf I} - {\bf VI} with the help of relations (\ref{form:cyclic relations 1}) - (\ref{form:cyclic relations 5}). This requires additional relations.

We will show how a part of these additional relations can be deduced from the relations {\bf I}, for the other sets of relations this is done similarly. As $(\vartheta^1)^3=0$, the relations {\bf I} can be written as follows
\begin{eqnarray}
\vartheta^1\vartheta^2\vartheta^2+\vartheta^2\vartheta^2\vartheta^1 &=& \vartheta^3\vartheta^1\vartheta^3,
\label{form:additional reltions 1}\\
\vartheta^2\vartheta^2\vartheta^1+\vartheta^2\vartheta^1\vartheta^2 &=& \vartheta^1\vartheta^3\vartheta^3,
\label{form:additional reltions 2}\\
\vartheta^2\vartheta^1\vartheta^2+\vartheta^1\vartheta^2\vartheta^2 &=& \vartheta^3\vartheta^3\vartheta^1.
\label{form:additional reltions 3}
\end{eqnarray}
For instance, making use of second relation (\ref{form:cyclic relations 2}), we can express
$$
\vartheta^1\vartheta^2\vartheta^2+\vartheta^2\vartheta^2\vartheta^1=-\vartheta^2\vartheta^1\vartheta^2,
$$
and, substituting this into (\ref{form:additional reltions 1}), we obtain
\be
\vartheta^2\vartheta^1\vartheta^2+\vartheta^3\vartheta^1\vartheta^3=0.
\label{form:extra relation}
\ee
Applying the same method to (\ref{form:additional reltions 2}),(\ref{form:additional reltions 3}), we get
\begin{eqnarray}
\vartheta^1\,\big((\vartheta^2)^2+(\vartheta^3)^2\big) &=& 0,\label{form:defining relations 1}\\
\big((\vartheta^2)^2+(\vartheta^3)^2\big)\,\vartheta^1 &=& 0.\label{form:defining relations 2}
\end{eqnarray}
It is easy to see that the relation (\ref{form:extra relation}) can be deduced from the relations (\ref{form:defining relations 1}),(\ref{form:defining relations 2}). Indeed, adding up relations (\ref{form:defining relations 1}),(\ref{form:defining relations 2}) and applying second relation in (\ref{form:cyclic relations 2}), we get (\ref{form:extra relation}). Similar calculations in the case of relations {\bf III}, {\bf V} give four more independent relations
\begin{eqnarray}
&& \vartheta^2\,\big((\vartheta^3)^2+(\vartheta^1)^2\big)=0,\;\;
    \big((\vartheta^3)^2+(\vartheta^1)^2\big)\,\vartheta^2=0,\\
&& \vartheta^3\,\big((\vartheta^1)^2+(\vartheta^2)^2\big)=0,\;\;
    \big((\vartheta^1)^2+(\vartheta^2)^2\big)\,\vartheta^3=0.
\end{eqnarray}
It is worth to note that relations {\bf II},{\bf IV}, {\bf VI} do not give new relations.

For the convenience of presenting the obtained algebraic relations, we arrange them in a following table
\begin{eqnarray}
&& \vartheta^1\,\big((\vartheta^2)^2+(\vartheta^3)^2\big)=0,\;\;\;
        \big((\vartheta^2)^2+(\vartheta^3)^2\big)\,\vartheta^1=0,\\
&&  \vartheta^2\,\big((\vartheta^3)^2+(\vartheta^1)^2\big)=0,\;\;\;
      \big((\vartheta^3)^2+(\vartheta^1)^2\big)\,\vartheta^2=0,\\
&&  \vartheta^3\,\big((\vartheta^1)^2+(\vartheta^2)^2\big)=0,\;\;\;
       \big((\vartheta^1)^2+(\vartheta^2)^2\big)\,\vartheta^3=0.
\end{eqnarray}
\begin{theorem}
An element $t(\vartheta)=t_{ijk}\vartheta^i\vartheta^j\vartheta^k$, where $t_{ijk}$ are components of a tensor of rank 3 and $\vartheta^1,\vartheta^2,\vartheta^3$ are generators of tensor algebra $T_\vartheta(E)$, is a solution of the equation $H^2\,t(\vartheta)=6\,t(\vartheta)$ if and only if the generators $\vartheta^1,\vartheta^2,\vartheta^3$ are subjected to the following ternary relations
\begin{eqnarray}
&&\{\vartheta^1,\vartheta^2,\vartheta^3\}=0,\;\;\;\{\vartheta^3,\vartheta^2,\vartheta^1\}=0,\label{form:cyclic relations 11}\\
&&\{\vartheta^1,\vartheta^1,\vartheta^2\}=0,\;\;\;\{\vartheta^2,\vartheta^2,\vartheta^1\}=0,\label{form:cyclic relations 22}\\
&&\{\vartheta^1,\vartheta^1,\vartheta^3\}=0,\;\;\;\{\vartheta^3,\vartheta^3,\vartheta^1\}=0,\label{form:cyclic relations 33}\\
&&\{\vartheta^2,\vartheta^2,\vartheta^3\}=0,\;\;\;\{\vartheta^3,\vartheta^3,\vartheta^2\}=0,\label{form:cyclic relations 44}\\
&&\;(\vartheta^1)^3=0,\;\;\;(\vartheta^2)^3=0,\;\;\;(\vartheta^3)^3=0,\label{form:cyclic relations 55}\\
&& \vartheta^1\,\big((\vartheta^2)^2+(\vartheta^3)^2\big)=0,\;\;\;
        \big((\vartheta^2)^2+(\vartheta^3)^2\big)\,\vartheta^1=0,\label{form:cyclic relations 66}\\
&&  \vartheta^2\,\big((\vartheta^3)^2+(\vartheta^1)^2\big)=0,\;\;\;
      \big((\vartheta^3)^2+(\vartheta^1)^2\big)\,\vartheta^2=0,\label{form:cyclic relations 77}\\
&&  \vartheta^3\,\big((\vartheta^1)^2+(\vartheta^2)^2\big)=0,\;\;\;
       \big((\vartheta^1)^2+(\vartheta^2)^2\big)\,\vartheta^3=0.\label{form:cyclic relations 88}
\end{eqnarray}
The ternary relations (\ref{form:cyclic relations 11}) - (\ref{form:cyclic relations 88}) are invariant under a representation of the rotations group in the vector space of tensors of rank 3. The subspace of vector space of tensors of rank 3, determined by the ternary relations (\ref{form:cyclic relations 11}) - (\ref{form:cyclic relations 88}), is 10-dimensional and a tensor representation of weight 2 of the rotations group in this 10-dimensional space can be split into two irreducible representations.
\end{theorem}
\begin{proof}
If we pass from one Cartesian coordinate system in Euclidean space $E$ to another by means of $g=(g^i_j)\in\mbox{SO}(3)$ then generators of tensor algebra undergo the transformation $\tilde\vartheta^i=g^i_j\vartheta^j$. Hence, in new coordinate system we have
$$
\{\tilde\vartheta^i,\tilde\vartheta^j,\tilde\vartheta^k\}=g^i_mg^j_ng^k_p\,\{\vartheta^m,\vartheta^n,\vartheta^p\}=0,
$$
which means that the relations (\ref{form:cyclic relations 11}) - (\ref{form:cyclic relations 55}) are invariant. In order to show that the relations (\ref{form:cyclic relations 66}) - (\ref{form:cyclic relations 88}) are also invariant we prove this only for the first relation in (\ref{form:cyclic relations 66}), an invariance of other relations of this type is proved analogously. First of all, we observe that the expression $(\vartheta^1)^2+\vartheta^2)^2+(\vartheta^3)^2$ is invariant under transformation $\vartheta^i\to g^i_j\vartheta^j$. Indeed we have
$$
\sum_{i=1}^3\,(\tilde\vartheta^i)^2=\delta_{ij}\tilde\vartheta^i\tilde\vartheta^j=
\delta_{ij}g^i_mg^j_n\,\vartheta^m\vartheta^n=\delta_{mn}\vartheta^m\vartheta^n=\sum_{i=1}^3(\vartheta^i)^2.
$$
Then
\begin{eqnarray}
\tilde\vartheta^1\;\big((\tilde\vartheta^2)^2+(\tilde\vartheta^3)^2\big) \!\!\!\!\!&=&\!\!\!\!
       \tilde\vartheta^1\;\big((\tilde\vartheta^1)^2+(\tilde\vartheta^2)^2+(\tilde\vartheta^3)^2\big)\nonumber\\
&=&\!\!\!\!(g^1_1\vartheta^1+g^1_2\vartheta^2+g^1_3\vartheta^3)\big((\vartheta^1)^2+(\vartheta^2)^2+(\vartheta^3)^2\big)\nonumber\\
&=&\!\!\!\!g^1_1\vartheta^1\big((\vartheta^2)^2+(\vartheta^3)^2\big)+
         g^1_2\vartheta^2\big((\vartheta^3)^2+(\vartheta^1)^2\big)\nonumber\\
&&\qquad\qquad\qquad\qquad\qquad +g^1_3\vartheta^3\big((\vartheta^1)^2+(\vartheta^2)^2\big)=0.\nonumber
\end{eqnarray}
To find the dimension of the subspace, determined by the relations (\ref{form:cyclic relations 11}) - (\ref{form:cyclic relations 88}), we subtract the number of relations in (\ref{form:cyclic relations 11}) - (\ref{form:cyclic relations 88}) from the number of all triple products of generators, i.e. 27-17=10. For instance the following triple products can be taken as a basis for this 10-dimensional subspace:
\begin{eqnarray}
&&\vartheta^1\vartheta^2\vartheta^3,\;\;\;\vartheta^2\vartheta^3\vartheta^1,\\
&&\vartheta^3\vartheta^2\vartheta^1,\;\;\;\vartheta^2\vartheta^1\vartheta^3,\\
&&\vartheta^1\vartheta^2\vartheta^2,\;\;\;\vartheta^2\vartheta^2\vartheta^1,\\
&&\vartheta^2\vartheta^3\vartheta^3,\;\;\;\vartheta^3\vartheta^3\vartheta^2,\\
&&\vartheta^3\vartheta^1\vartheta^1,\;\;\;\vartheta^1\vartheta^1\vartheta^3.
\end{eqnarray}
The statement that in this 10-dimensional subspace there is a double irreducible representation of the rotations group follows from the theory of representations of the rotations group \cite{Gelfand:2018}.
\end{proof}
One may notice that the relations (\ref{form:cyclic relations 11}) - (\ref{form:cyclic relations 88}) can naturally be divided into two parts. In the first part (\ref{form:cyclic relations 11}) - (\ref{form:cyclic relations 55}), relations can not be obtained from quadratic (or binary relations) for generators. However, the relations in the second part (\ref{form:cyclic relations 66}) - (\ref{form:cyclic relations 88}) can be derived from binary relations, moreover, one binary relation is enough for this. In the proof of the theorem we used the quadratic expression
$$
(\vartheta^1)^2+(\vartheta^2)^2+(\vartheta^3)^2=\delta_{ij}\vartheta^i\vartheta^j,
$$
and showed that it is invariant with respect to rotations of the space $E$. It is easy to see that all six relations
(\ref{form:cyclic relations 66}) - (\ref{form:cyclic relations 88}) follow from the quadratic relation
\be
(\vartheta^1)^2+(\vartheta^2)^2+(\vartheta^3)^2=0,
\ee
if we successively multiply it one time from the left and second time from the right by generators $\vartheta^1,\vartheta^2,\vartheta^3$ and use $(\vartheta^1)^3=(\vartheta^2)^3=(\vartheta^3)^3=0$. Obviously, the structure of the 10-dimensional space defined by the relations (\ref{form:cyclic relations 11}) - (\ref{form:cyclic relations 88}) does not change, that is, it will be the space of solutions of equation (\ref{form:particular case of basic equation l=2}). However, the structure of the subspace spanned by the binary products of the generators $\vartheta^i\vartheta^j$ will change, its dimension will decrease by one.

Now we can give the following definition:
\begin{defn}
\emph{An algebra associated with a representations of weight 2 of the rotations group} in the space of rank 3 tensors of 3-dimensional Euclidean space $E$ is a unital associative algebra over $\mathbb C$ generated by $\vartheta^1,\vartheta^2,\vartheta^3$, which satisfy the following conditions:
\begin{enumerate}
\item
rotation $e_i\to g^j_i\,e_j$ of the space $E$ implies the transformation of generators according to the formula $\vartheta^i\to \tilde g^i_j\,\vartheta^j$, where $\tilde g$ is inverse matrix for $g$,
\item
$\delta_{ij}\vartheta^i\vartheta^j=0$ (quadratic relations),
\item
$\vartheta^i\vartheta^j\vartheta^k+\vartheta^j\vartheta^k\vartheta^i+\vartheta^k\vartheta^i\vartheta^j=0$ (cubic relations),
\end{enumerate}
where $i,j,k$ is any triple of integers $1,2,3$.
This algebra will be denoted by $\frak T_\vartheta(E)$ and its 10-dimensional subspace spanned by triple products of generators, which obey the cubic relations (\ref{form:cyclic relations 11}) - (\ref{form:cyclic relations 88}), will be denoted by $\frak T_\vartheta^3(E)$.
\end{defn}
\section{Algebras associated with irreducible tensor representations of weight 2 of $\mbox{SO}(3)$}
We have the algebra $\frak T_\vartheta(E)$ associated with tensor representations of the rotations group of weight 2 and in the 10-dimensional subspace $\frak T^3_\vartheta(E)$ spanned by triple products of generators we have a double irreducible representation of $\mbox{SO}(3)$. In \cite{Gelfand:2018} it is explained that any splitting of 10-dimensional space $\frak T^3(E)$ into direct sum of subspaces, which are invariant under representation of $\mbox{SO}(3)$, gives two irreducible tensor representations of $\mbox{SO}(3)$. There are few ways to split invariantly $\frak T^3(E)$ into subspaces and one of them is to use the subspaces of eigenvectors of linear operators induced by substitutions.

In general this method can be described as follows. Assume we consider a vector space of rank $r$ tensors $T^r(E)$ and let $\sigma$ be a substitution of integers $1,2,\ldots,r$. Then it induces the linear operator $L_\sigma$ in a vector space $T^r(E)$ defined by
\be
L_\sigma(t_{i_1i_2\ldots i_r})=t_{i_{\sigma(1)}i_{\sigma(2)}\ldots i_{\sigma(r)}}.
\ee
For any substitution there exists an integer $n$ such that $\sigma^n=e$, where $e$ is the identical substitution. Thus, $L_{\sigma}^n=I$, where $I$ is identical linear transformation in $T^r(E)$. From this it follows that the eigenvalues of $L_{\sigma}$ are $n$th roots of unity. Let $q$ be a primitive $n$th root of unity, then $1,q,q^2,\ldots,q^{n-1}$ are all $n$th roots of unity. For any $n$th root of unity $q^k$ the equation
\be
L_\sigma(t_{i_1i_2\ldots i_r})=q^k\;t_{i_1i_2\ldots i_r},
\label{form:equation for eigenvectors q_k}
\ee
determines the subspace of eigenvectors of $L_\sigma$ and the direct sum of all these subspace gives a whole vector space $T^r(E)$. Evidently this splitting into direct sum of subspaces is invariant under tensor representation of the rotations group.

We can extend this method to a subspace $T^r_\vartheta(E)$ of tensor algebra $T_\vartheta(E)$ as follows. Multiplying both sides of the equation (\ref{form:equation for eigenvectors q_k}) by $\vartheta^{i_1}\,\vartheta^{i_2}\ldots \vartheta^{i_r}$ and summing up with respect to all indices, we get
$$
L_\sigma(t_{i_1i_2\ldots i_r})\vartheta^{i_1}\vartheta^{i_2}\ldots\vartheta^{i_r}=q^k\;t_{i_1i_2\ldots i_r}
                                                                      \vartheta^{i_1}\vartheta^{i_2}\ldots\vartheta^{i_r},
$$
or
\be
t_{i_{\sigma(1)}i_{\sigma(2)}\ldots i_{\sigma(r)}}\vartheta^{i_1}\vartheta^{i_2}\ldots\vartheta^{i_r}=q^k\;
                 \;t_{i_1i_2\ldots i_r}\vartheta^{i_1}\vartheta^{i_2}\ldots\vartheta^{i_r}.
\label{form:rearranged equation}
\ee
But the terms at the left hand side of (\ref{form:rearranged equation}) can be rearranged as follows
$$
t_{i_{\sigma(1)}i_{\sigma(2)}\ldots i_{\sigma(r)}}\vartheta^{i_1}\vartheta^{i_2}\ldots\vartheta^{i_r}=
   t_{i_1i_2\ldots i_r}\vartheta^{i_{\varsigma(1)}}\vartheta^{i_{\varsigma(2)}}\ldots \vartheta^{i_{\varsigma(r)}},
$$
where $\varsigma$ is the inverse substitution for $\sigma$. Hence equation (\ref{form:rearranged equation}) can be written in the form
\be
t_{i_1i_2\ldots i_r}\;\vartheta^{i_{\varsigma(1)}}\vartheta^{i_{\varsigma(2)}}\ldots \vartheta^{i_{\varsigma(r)}}=
    t_{i_1i_2\ldots i_r}\;q^k\,\vartheta^{i_1}\vartheta^{i_2}\ldots\vartheta^{i_r}.
\ee
Assuming the components of a tensor $t$ to be independent, we can determine the subspace of a vector space $T^r_\vartheta(E)$ by imposing additional conditions on the $r$th order products of generators
\be
\vartheta^{i_{\varsigma(1)}}\vartheta^{i_{\varsigma(2)}}\ldots \vartheta^{i_{\varsigma(r)}}=q^k\;
                    \vartheta^{i_1}\vartheta^{i_2}\ldots\vartheta^{i_r},
\ee
or
\be
\vartheta^{i_1}\vartheta^{i_2}\ldots\vartheta^{i_r}=q^{-k}
                \vartheta^{i_{\varsigma(1)}}\vartheta^{i_{\varsigma(2)}}\ldots \vartheta^{i_{\varsigma(r)}},
\label{form:algebraic relations with q^k}
\ee
and this method is equivalent to the method based on equation (\ref{form:equation for eigenvectors q_k}). It should be noted an important thing here is that the algebraic relations (\ref{form:algebraic relations with q^k})  are invariant with respect to a representation of the rotations group.

Now we apply this method to the 10-dimensional vector space $\frak T_\vartheta^3(E)$ of the algebra $\frak T_\vartheta(E)$. Let $\sigma$ be the substitution of integers 1,2,3 defined by
$$
\sigma=\left(
    \begin{array}{ccc}
      1 & 2 & 3 \\
      2 & 3 & 1 \\
    \end{array}
  \right).
$$
Obviously $\sigma^3=e$. Let $q=\exp{(2\pi\,i/3)}$ be the primitive cubic root of unity and $\bar q=q^2$ its conjugate, then $1+q+\bar q=0$. Previous considerations lead to introduction of the following algebraic structure:
\begin{defn}
\emph{A $q$-algebra \big($\bar q$-algebra\big) associated with an irreducible tensor representation of the rotations group of weight 2} in 3-dimensional space $E$ is a unital associative algebra over $\mathbb C$ generated by $\theta^1,\theta^2,\theta^3\;\;$ \big($\bar\theta^1,\bar\theta^2,\bar\theta^3$\big), which satisfy the following conditions:
\begin{enumerate}
\item rotation $e_i\to g^j_i\,e_j$ of Euclidean space $E$ entails the transformation \newline
$\theta^i\to (g^{-1})^i_j\,\theta^j,\qquad\qquad\qquad\;\big(\bar\theta^i\to (g^{-1})^i_j\,\bar\theta^j\big),$
\item
$\delta_{ij}\theta^i\theta^j=0,\;\;\;\;\;\;\qquad\qquad\qquad \big(\delta_{ij}\bar\theta^i\bar\theta^j=0\big)$,
\item
$\theta^{i_1}\theta^{i_2}\theta^{i_3}
              =q\;\theta^{i_{\sigma(1)}}\vartheta^{i_{\sigma(2)}}\theta^{i_{\sigma(3)}},\;
              \big(\bar\theta^{i_1}\bar\theta^{i_2}\bar\theta^{i_3}
                   =\bar q\;\bar\theta^{i_{\sigma(1)}}\bar\theta^{i_{\sigma(2)}}\bar\theta^{i_{\sigma(3)}}\big)$,
\end{enumerate}
where $i_1,i_2,i_3$ is any triple of integers $1,2,3$.
The $q$-algebra $\big(\bar q\mbox{-algebra}\big)$ will be denoted by $\frak T_{\theta}(E)$
$\big(\frak T_{\bar\theta}(E)\big)$ and its 5-dimensional subspace spanned by triple products of generators will be denoted by
$\frak T_{\theta}^3(E)$ $\big(\frak T_{\bar\theta}^3(E)\big)$.
\end{defn}
We describe the structure of a $q$-algebra $\frak T_\theta(E)$, because $\bar q$-algebra $\frak T_{\bar\theta}(E)$ has a similar structure.
As previously noted, unital associative algebra generated by $\theta^1,\theta^2,\theta^3$, which obey only cubic relations
\begin{equation}
\theta^{i_1}\theta^{i_2}\theta^{i_3}
              =q\;\theta^{i_{\sigma(1)}}\vartheta^{i_{\sigma(2)}}\theta^{i_{\sigma(3)}},
\label{form:cubic relations with q}
\end{equation}
(and there are no quadratic relations!) was studied in a series of papers and its structure is known. In particular, it can be shown that this algebra does not contain monomials whose degree is greater than or equal to four. Indeed, sequentially applying the associativity and cubic relations (\ref{form:cubic relations with q}) to the triples of adjacent generators in the product of four generators, after a series of transformations, we return to the original product, but with a coefficient different from unity, which implies that any such product must be equal to zero. Indeed, we have
\begin{eqnarray}
&&(\theta^i\theta^j\theta^k)\theta^l=q\;(\theta^j\theta^k\theta^i)\theta^l=q\;\theta^j(\theta^k\theta^i\theta^l)=
     q^2\;\theta^j(\theta^i\theta^l\theta^k)
 =q^2\;(\theta^j\theta^i\theta^l)\theta^k\nonumber\\
   &&\qquad\quad  =q^3\;(\theta^i\theta^l\theta^j)\theta^k=q^3\;\theta^i(\theta^l\theta^j\theta^k)=
           q^4\;\theta^i(\theta^j\theta^k\theta^l)=q\;(\theta^i\theta^j\theta^k)\theta^l.\nonumber
\end{eqnarray}
Obviously, this result also holds in the case of $q$-algebra $\frak T_\theta(E)$, since the cubic relations (\ref{form:cubic relations with q}) consist a part of the system of defining relations of $q$-algebra $\frak T_\theta(E)$. Thus, $q$-algebra $\frak T_\theta(E)$ is a direct sum of the following vector spaces: the one-dimensional space generated by the identity element of the algebra, the 8-dimensional space generated by pairwise products $\theta^i\theta^j$ (there will be 9 such products, but we have the quadratic relation
$$
\delta_{ij}\theta^i\theta^j=0,
$$
which reduces the dimension to 8), and the 5-dimensional space generated by triple products $\theta^i\theta^j\theta^k$. We are particularly interested in the latter, since it is a representation space for irreducible representation of the rotations group $\mbox{SO}(3)$ of weight 2. The following interesting analogy should be noted. In the Grassmann algebra of three-dimensional space, the highest degree of a monomial equals three, that is, all monomials of a higher degree vanish. On the other hand, the space of monomials of the third degree is a one-dimensional representation space for an irreducible representation of the group of rotations of weight zero in the spaces of skew-symmetric tensors of the third rank.

First of all, we choose a basis consisting of independent monomials for this 5-dimensional vector space. The next monomials are independent
\begin{equation}
\theta^1\theta^2\theta^2,\;\theta^2\theta^3\theta^3,\;\theta^3\theta^1\theta^1,\;
\theta^1\theta^2\theta^3,\;\theta^3\theta^2\theta^1.
\label{form:independent monomials }
\end{equation}
Now in this space it is necessary to define a Hermitian scalar product in such a way that it corresponds to a unitary representation of the rotations group, that is, for each skew-symmetric matrix $F\in\mbox{so}(3)$ the operator $dR_F$ generated by it would be traceless skew-Hermitian. We do this using the orthonormal basis $\eta^A$, where superscript $A$ runs from 1 to 5, which is defined as follows
\begin{eqnarray}
\eta^1 &=& \sqrt{2}\;\theta^1\theta^2\theta^2,\nonumber\\
\eta^2 &=& \sqrt{2}\;\theta^2\theta^3\theta^3,\nonumber\\
\eta^3 &=& \sqrt{2}\;\theta^3\theta^1\theta^1,\nonumber\\
\eta^4 &=&            \theta^1\theta^2\theta^3,\nonumber\\
\eta^5 &=&             \theta^3\theta^2\theta^1.\nonumber
\end{eqnarray}
Hence we define the Hermitian scalar product in the 5-dimensional space spanned by triples products of $q$-algebra $\frak T_\theta(E)$ by means of
$$
<\eta^A,\eta^B>=\delta^{AB}.
$$
Applying (\ref{form:formula for dR_F}) to basic monomials $\eta^A$, for any $F\in\mbox{so}(3)$ we find
\begin{eqnarray}
dR_F(\eta^1) &=& -F^1_2\,\eta^2-F^1_3\,\eta^3+\sqrt{2}\,F^2_3\,\eta^4+q\sqrt{2}\,\,F^2_3\,\eta^5,\nonumber\\
dR_F(\eta^2) &=& -F^2_1\,\eta^1-F^2_3\,\eta^3-\sqrt{2}\,F^1_3\,\eta^4-q^2\sqrt{2}\,F^1_3\,\eta^5,\nonumber\\
dR_F(\eta^3) &=& -F^3_1\,\eta^2-F^3_2\,\eta^2+q\sqrt{2}\,F^1_2\,\eta^4+\sqrt{2}\,F^1_2\,\eta^5,\nonumber\\
dR_F(\eta^4) &=& -\sqrt{2}\,F^2_3\,\eta^1-q\sqrt{2}F^3_1\,\eta^2-q^2\sqrt{2}\,F^1_2\,\eta^3,\nonumber\\
dR_F(\eta^5) &=& -q^2\sqrt{2}\,F^2_3\,\eta^1-q\sqrt{2}F^3_1\,\eta^2-\sqrt{2}\,F^1_2\,\eta^3,\nonumber
\end{eqnarray}
which gives the matrix of an operator $dR_F$
\begin{equation}
\left(
                \begin{array}{cccccc}
                  0 & F^1_2 & F^1_3 &|& -\sqrt{2}\,F^2_3 & -\bar q\sqrt{2}\,F^2_3 \\
                  F^2_1 & 0 & F^2_3 &|& -q\sqrt{2}\,F^3_1 & -q\sqrt{2}\,F^3_1 \\
                  F^3_1 & F^3_2 & 0 &|& -\bar q\sqrt{2}\,F^1_2 & -\sqrt{2}\,F^1_2 \\
                  --&--&--&--&--&--\\
                  \sqrt{2}\,F^2_3 & \bar q\sqrt{2}\,F^3_1 & q\sqrt{2}\,F^1_2 &|& 0 & 0 \\
                  q\sqrt{2}\,F^2_3 & \bar q\sqrt{2}\,F^3_1 & \sqrt{2}\,F^1_2 &|& 0 & 0 \\
                \end{array}
              \right)
\label{form:matrix dR_F}
\end{equation}
For any $F\in\mbox{so}(3)$ this matrix is traceless and skew-Hermitian, i.e.
$$
\mbox{Tr}(dR_F)=0,\;\;dR_F+(dR_F)^\dag=0,
$$
where $\dag$ stands for complex conjugation and transposition.
Therefore matrix (\ref{form:matrix dR_F}) gives an explicit formula for the homomorphism of two Lie algebras $dR:\mbox{so}(3)\to\mbox{su}(5)$.

It is interesting to note that the submatrix, which is located in the upper left corner of matrix (\ref{form:matrix dR_F}), is an infinitesimal rotation $F$ in three-dimensional space. Thus, in the 5-dimensional space of irreducible representation of $\mbox{SO}(3)$ spanned by independent triple monomials, it is natural to single out a three-dimensional subspace, whose basis vectors are monomials
\begin{equation}
\eta^1 = \sqrt{2}\;\theta^1\theta^2\theta^2,\;\;
\eta^2 = \sqrt{2}\;\theta^2\theta^3\theta^3,\;\;
\eta^3 = \sqrt{2}\;\theta^3\theta^1\theta^1.
\label{form:three first independent monomials}
\end{equation}
An orthogonal complement to this subspace is the two-dimensional subspace spanned by the monomials
\begin{equation}
\eta^4 = \theta^1\theta^2\theta^3,\;\;\eta^5 = \theta^3\theta^2\theta^1.
\label{form:two last independent monomials}
\end{equation}
Obviously, if we consider a vector of the three-dimensional space (\ref{form:three first independent monomials}), apply the matrix of a linear operator $dR_F$ to it, and then project the resulting five-dimensional vector onto the same three-dimensional subspace by orthogonal projection, we obtain the infinitesimal rotation of an initial vector.
\section{Discussion}
Grassmann algebra or exterior algebra of three-dimensional space can be constructed as follows. In three-dimensional space, we choose a basis and this choice determines the basis of dual space. As generators of an algebra we take the covectors of the dual basis and define their exterior multiplication (wedge product), which is associative and antisymmetric. The space spanned by triple products of covectors of dual basis is one-dimensional and it is the space of an irreducible representation of the group of rotations $\mbox{SO}(3)$ of weight zero in the space of skew-symmetric rank 3 tensors.

But the rotations group has, besides the weight 0 representation, representations of the weight 1,2 and 3 in the space of rank 3 tensors. In the present paper, we show that representations of weight 2 can be used to construct an analog of Grassmann algebra. We construct two algebras, these algebras are called $q$-algebra and $\bar q$-algebra, which can be considered as analogues of a Grassmann algebra. In these algebras, as well as in Grassmann algebra, the highest order of a monomial that can be formed from generators is 3, that is, monomials of higher orders vanish. The space spanned by independent third-order monomials is five-dimensional. This five-dimensional space is the space of an irreducible representation of the rotations group in a space of rank 3 tensors. At an infinitesimal level, we obtain a representation of the Lie algebra of the rotations group $\mbox{so}(3)\to\mbox{so}(5)$ and it can be shown that at the group level we have $\mbox{SO}(3)\to\mbox{SO}(5)$. Since group $\mbox{SO}(5)$  is the gauge group of Georgi-Glashow model, we think that a $q$-algebra, as well as $\bar q$-algebra, can have applications in this theory. The reason for this is the explicit form of the matrix $dR_F$ obtained in the present paper
\begin{equation}
\left(
                \begin{array}{cccccc}
                  0 & F^1_2 & -F^3_1 &|& -\sqrt{2}\,F^2_3 & -\bar q\sqrt{2}\,F^2_3 \\
                  -F^1_2 & 0 & F^2_3 &|& -q\sqrt{2}\,F^3_1 & -q\sqrt{2}\,F^3_1 \\
                  F^3_1 & -F^2_3 & 0 &|& -\bar q\sqrt{2}\,F^1_2 & -\sqrt{2}\,F^1_2 \\
                  --&--&--&--&--&--\\
                  \sqrt{2}\,F^2_3 & \bar q\sqrt{2}\,F^3_1 & q\sqrt{2}\,F^1_2 &|& 0 & 0 \\
                  q\sqrt{2}\,F^2_3 & \bar q\sqrt{2}\,F^3_1 & \sqrt{2}\,F^1_2 &|& 0 & 0 \\
                \end{array}
              \right).
\label{form:discussion_matrix dR_F}
\end{equation}
The standard model quarks and leptons fit nicely into representations of $\mbox{SO}(5)$. Each left-handed fermions generation of standard model can be grouped into two representations $\bf\bar 5$ and $\bf 10$ of $\mbox{SO}(5)$ in the following way
$$
\bf\bar 5=\left(
          \begin{array}{c}
            d^c_1 \\
            d^c_2 \\
            d^c_3 \\
            ---\\
            e^{-} \\
            -\nu_e \\
          \end{array}
        \right),
{\bf 10}=\left(
  \begin{array}{cccccc}
    0 & u^c_3 & -u^c_2 &|& -u^1 & -d^1 \\
    -u^c_3 & 0 & u^c_1 &|& -u^2 & -d^2 \\
    u^c_2 & -u^c_1 & 0 &|& -u^3 & -d^3 \\
    --&--&--&--&--&--\\
    u^1 & u^2 & u^3 &|& 0 & -e^c \\
    d^1 & d^2 & d^3 &|& e^c & 0 \\
  \end{array}
\right)
$$
A comparison of two matrices, that is, matrix $dR_F$ with matrix 10, shows that a $q$-algebra can be applied to the description of quarks, because in the lower right corner of matrix $dR_F$ there is a completely zero block, while in matrix 10 this block is responsible for the electron.

We would also like to note possible mathematical applications of the constructed $q$-algebra. As is known, Grassmann algebra serves as the basis for constructing the algebra of differential forms. This is equivalent to passing from tensors to tensor fields. The differential form cannot contain the square of any differential, because the generators of Grassmann algebras anticommute. In the calculus of differential forms, this leads to the fact that the exterior differential $d$ gives only first-order derivatives of tensors fields, second-order derivatives disappear due to the property $d^2=0$. The monomials of $q$-algebra $\frak T_\theta(E)$ can contain squares of generators and this indicates that it will be possible to construct by means of this algebra an analog of differential forms in which the differential (analog of exterior differential) will give second-order derivatives and this calculus could be used to study operators that include second-order derivatives, for example, the Laplace operator.
\providecommand{\href}[2]{#2}

\address{
Institute of Mathematics and Statistics, University of Tartu\\
Narva mnt 18, 51009 Tartu, Estonia\\
\email{viktor.abramov@ut.ee}\\
}

\end{document}